\documentclass[twocolumn,pre]{revtex4}
\usepackage{amsmath,amsfonts,amssymb,amsthm}
\usepackage[usenames]{color}

\usepackage{bm}

\usepackage{xspace}

\usepackage{paralist}

\usepackage{graphicx}
\usepackage{subfig}

\usepackage[colorlinks,urlcolor=blue,citecolor=blue,linkcolor=blue]{hyperref}

\usepackage{braket}

\usepackage{algorithm}
\usepackage{algpseudocode}
\usepackage{setspace}

\newtheorem{theorem}{Theorem}
\newtheorem{claim}[theorem]{Claim}

\newtheorem{example_hidden}[theorem]{Example}

\newcommand\blackfbox[1]{\textcolor{black}{\setlength\fboxsep{5pt}\fbox{\textcolor{black}{#1}}}}

\newcommand{\cE}{{\cal E}}

\newcommand{\nint}[1]{\lfloor #1\rceil} %
\newcommand{\hd}{\hat{d}}
\newcommand{\EX}{\hbox{\bf E}}
\newcommand{\Sec}[1]{\hyperref[sec:#1]{\S\ref*{sec:#1}}} %
\newcommand{\Eqn}[1]{\hyperref[eq:#1]{(\ref*{eq:#1})}} %
\newcommand{\Fig}[1]{\hyperref[fig:#1]{Fig.\,\ref*{fig:#1}}} %
\newcommand{\Tab}[1]{\hyperref[tab:#1]{Tab.\,\ref*{tab:#1}}} %
\newcommand{\Thm}[1]{\hyperref[thm:#1]{Thm.\,\ref*{thm:#1}}} %
\newcommand{\Lem}[1]{\hyperref[lem:#1]{Lem.\,\ref*{lem:#1}}} %
\newcommand{\Prop}[1]{\hyperref[prop:#1]{Prop.~\ref*{prop:#1}}} %
\newcommand{\Cor}[1]{\hyperref[cor:#1]{Cor.~\ref*{cor:#1}}} %
\newcommand{\Def}[1]{\hyperref[def:#1]{Defn.~\ref*{def:#1}}} %
\newcommand{\Alg}[1]{\hyperref[alg:#1]{Alg.~\ref*{alg:#1}}} %
\newcommand{\Ex}[1]{\hyperref[ex:#1]{Ex.~\ref*{ex:#1}}} %
\newcommand{\Clm}[1]{\hyperref[clm:#1]{Claim~\ref*{clm:#1}}} %

\newcommand{\GraphFig}[3]{
  \begin{figure*}[t]
    \centering
    \subfloat{\label{fig:#1-dd}
      \includegraphics[width=.3\textwidth,trim=0 0 0 0]{#1-dd}
    }
    \subfloat{\label{fig:#1-cc}
      \includegraphics[width=.3\textwidth,trim=0 0 0 0]{#1-cc}
    }
    \subfloat{\label{fig:#1-eigs}
      \includegraphics[width=.3\textwidth,trim=0 0 0 0]{#1-eigs}
    }
    \caption{Properties of #1, #2, compared with the BTER and CL
      models. #3}
    \label{fig:#1}
  \end{figure*}
}

\newcommand{\SG}{S} %

\begin{document}

\title{Community structure and scale-free collections of Erd\H{o}s--R\'enyi graphs}

\author{C. Seshadhri, Tamara G. Kolda, Ali Pinar}
\affiliation{Sandia National Laboratories, Livermore, CA, USA}

\date{\today}

\begin{abstract}
Community structure plays a significant role in the analysis of social
networks and similar graphs, yet this structure is little understood
and not well captured by most models. 
We formally define a community to be a subgraph that is
internally highly connected and has no deeper substructure.
We use tools of combinatorics to show that any such community must
contain a dense Erd\H{o}s--R\'enyi (ER) subgraph.
Based on mathematical arguments, we hypothesize that any graph 
with a heavy-tailed degree distribution and community structure 
must contain a scale free collection of dense ER subgraphs. 
These theoretical observations corroborate well with
empirical evidence.
From this, we propose the Block Two-Level Erd\H{o}s--R\'enyi (BTER) model,
and demonstrate that it accurately captures the observable
properties of many real-world social networks.
\end{abstract}

\maketitle

\section{Introduction}
\label{sec:intro}

Graph analysis is becoming increasingly prevalent in the quest to
understand diverse phenomena like social relationships, scientific collaboration, purchasing
behavior, computer network traffic, and more. We refer
to graphs coming from such scenarios collectively as \emph{interaction networks}.
A significant amount of investigation has been done to understand the 
graph-theoretic properties common to interaction networks. 
Of particular importance is the notion of community structure. 
Interaction networks typically
decompose into internally well-connected sets referred
to as low conductance or high modularity cuts \cite{GiNe02, New06}.
Moreover, many graphs have high clustering coefficients \cite{WaSt98},
which is indicative of underlying community structure.  
Communities occur in a variety of sizes, though the largest community is 
often much smaller than the graph itself \cite{LeLaDa08, LaKiSa10}.
Community  analysis can reveal important patterns, decomposing large
collections of interactions into more meaningful components.

\subsection{A Theory of Communities} \label{sec:comm}

One metric of the quality of a community is the modularity metric
\cite{New06}. 
There are other measures such as conductance \cite{Ch92-book}, but they are equivalent to modularity in terms of our
intentions.
Consider a graph $G$ (undirected) with $n$ vertices and
degrees $d_1, d_2, \ldots, d_n$. Let $m = \frac{1}{2}\sum_{i=1}^n d_i$ denote
the number of edges. We say a subgraph $\SG$ has high modularity if $\SG$
contains many more internal edges than predicted by a \emph{null} model,
which says vertices $i$ and $j$ are connected with probability
$d_id_j/2m$. (Technically, the probability is $\min(1,
  d_id_j/2m)$, but we keep the notation simple for clarity.)
We refer to the null model as the CL model, based on its formalization
by Chung and Lu \cite{ChLu02, ChLu02-2}; see also Aiello et al.\@
\cite{AiChLu01} and the edge-configuration model of Newman et al.\@
\cite{NeWaSt02}.

Given a high modularity subgraph $\SG$, we say it is a \emph{module} if it
does not contain any further substructures of interest; in other
words, it is internally well-modeled by CL.
Formally, assume $\SG$ has $r$ nodes with \emph{internal}
degrees $\hd_1, \hd_2, \ldots, \hd_{r}$ and let the number of
edges in $\SG$ be denoted by $s = \frac{1}{2}\sum_{i=1}^r \hd_i$.
Consider the CL model on $S$, where edge $(i,j)$ occurs with probability
$\hd_i\hd_j/2s$. We call $S$ a \emph{module} if the induced subgraph on $S$ 
(the subgraph internal to $S$) is modeled well by this CL model.
Looking at the contrapositive,
if $\SG$ is not a module, then $\SG$ itself contains a subset of vertices that should
be separated out. A module can be thought of as an ``atomic" substructure within a graph.
In this language, we can think of community detection algorithms
as breaking a graph into modules. This discussion is not complete, however,
since communities are not just modules, but also internally well-connected.

Interaction networks have an abundance of triangles,
a fact that Watts and Strogatz \cite{WaSt98} succinctly express
through \emph{clustering coefficients}. Barrat and Weigt \cite{BaWe00} defined this as
\begin{equation}
  \label{eq:c}
  C = \frac{3 \times \text{total number of triangles}}
  {\text{total number of wedges}},
\end{equation}
where a \emph{wedge} is a path of length 2 \cite{WaSt98,GiNe02}.
It has been observed that $C$ ``has typical values in the range of 0.1
to 0.5 in many real-world networks'' \cite{GiNe02}. Moreover, our own studies have revealed 
that the node-level clustering coefficient (first used in \cite{WaSt98}), $C_i$, defined by 
\begin{equation}
  \label{eq:lc}
  C_i = \frac{\text{number of triangles incident to node } i}
  {\text{number of wedges centered at node } i},
\end{equation}
is typically highest for small degree nodes. 
Large clustering
coefficients are considered a manifestation of the community structures. 
Naturally, we expect the triangles to be largely contained within the
communities due to their high internal connectivity.

We now formally define a \emph{community} to be a module with a large
internal clustering coefficient. 
More formally, we say a module is a community if the expected
number of triangles is more than $(\kappa/3)\sum_i {\hd_i \choose 2}$,
for some constant $\kappa$.
In other words, a community is tightly connected internally and has few
external links.  A graph has \emph{community structure} if it (or at
least a constant fraction of it) can be broken up into communities.
The benefit of this formalism is that we can now try to understand
what graphs with community structure look like. 

Let us first begin by just focusing on a single community.
It seems fairly intuitive that a community cannot be large while
comprising only low
degree vertices nor that it consists of a single high-degree node
connected to degree-one vertices (a star). We
can actually prove a structural theorem about a community, given our formalization.
Recall that an Erd\H{o}s--R\'enyi (ER) graph
\cite{ErRe59,ErRe60} on $n$ vertices with connection probability $p$
is a graph such that each 
pair of vertices is independently connected with probability $p$. If $p$ is
a constant, we call this a dense ER graph; if $p = O(1/n)$, then
we call this a sparse ER graph. Using
triangle bounds from extremal combinatorics and some probabilistic arguments,
we can prove the following theorem.

\begin{theorem} \label{thm:community} A constant fraction of the edges in a community 
are contained in a dense Erd\H{o}s--R\'enyi graph. More formally, if the community
has $s$ edges, then there must be $\Omega(\sqrt{s})$ vertices with degree $\Omega(\sqrt{s})$.
\end{theorem}

This theorem is interesting because even though it is well known
that ER graphs are not good models
for interaction networks, they nonetheless form an important building
block for the communities. 
We interpret this theorem as saying that the simplest possible community is just
a dense ER graph. Building on this simple intuition, we think of an interaction
network as consisting of a large collection of dense ER graphs. 

This leads naturally to a question about
the distribution of sizes of these ER components. For that, consider
the power law degree distribution observed by Barab\'{a}si and Albert
\cite{BaAl99} and others.
They show that interaction graphs 
exhibit heavy-tailed degree distributions such as
\begin{equation}
  \label{eq:powerlaw}
  X_d \propto d^{-\gamma}
\end{equation}
where $X_d$ is the number of nodes of degree $d$ and $\gamma$ is the
power law exponent. 

Suppose we packed nodes with a heavy-tailed degree distribution
into a collection of dense ER graphs. A community of with
$s$ edges would be a dense ER graph of $\sqrt{s}$ vertices of degree $\sqrt{s}$, so
the size (in vertices) of the community is exactly $\sqrt{s}$.
Setting $d = \sqrt{s}$, the number of such communities is proportional to 
$$\frac{n/d^{\gamma}}{d} = \frac{n}{d^{\gamma+1}}.$$
This forms a scale-free distribution of communities, exactly as observed 
by many studies on community structure \cite{LeLaDa08,LaKiSa10}. \emph{Hence,
we hypothesize that real-world interaction networks consist of a scale-free
collection of dense Erd\H{o}s--R\'enyi graphs.} This is consistent with most
of the important observed properties of these networks.

Our analysis immediately leads to connections with Dunbar's celebrated result on 
``mean group sizes" of humans (reported to be around 148 with 95\% confidence limits of 100--231).
Empirically, this has been reported by a variety
of studies \cite{HeVi09,LaKiSa10,LeLaDa08, GoPeVe11}. If there exists a community of size $d$,
it must satisfy $n/d^{\gamma+1} \geq 1$. For the maximum community size $\bar{d}$, we have $\bar{d} \approx n^{1/(\gamma+1)}$.
For $n$ being a million and $\gamma = 2$, we get an estimate for a $100$, surprisingly close to Dunbar's
estimate.

As an aside, \Thm{community} also
proves that CL by itself is not a good model for interaction networks. Suppose
the entire graph $G$ (with $m$ edges) can be modeled as a CL graph. Since $G$ has a high clustering coefficient,
then $G$ itself is a module. Hence, $G$ must have $\Omega(\sqrt{m})$ vertices with degree $\Omega(\sqrt{m})$, but this violates
the tail behavior of the degree distribution.

\subsection{The BTER model} \label{sec:bter}

Based on the idea of a graph comprising ER communities, we propose the Block
Two-Level Erd\H{o}s-R\'{e}nyi model (BTER).
The advantages of the BTER model are that it has community structure in the
form of dense ER subgraphs and that it matches well with real-world graphs.
We briefly describe the model here and provide a more detailed
explanation and comparisons to real-world graphs in subsequent sections.

The first phase (or level) of BTER builds a collection
of ER blocks in such a way that the specified degree distribution is respected.
The BTER model allows one to construct a graph with \emph{any} degree distribution.
Real-world degree distributions might be idealized as power laws, but it
is by no means a completely accurate description \cite{ClShNe09,SaGaRoZh11}.
When the degree distribution is heavy tailed, then the BTER graph  
naturally has scale-free ER subgraphs.
The internal connectivity of the ER graphs is specified
by the user and can be tuned to match observed data.

The second phase of BTER interconnects the blocks. We assume that each node
has some \emph{excess degree} after the first phase. For example, if vertex $i$ should
have $d_i$ incident edges (according to the input degree distribution), and it has
$d'_i$ edges from its ER block, then the excess degree is $d_i - d'_i$.
We use a CL model (which can be considered as a weighted form of ER)
over the excess degrees to form the edges that connect communities.

\subsection{Previous models} \label{sec:prev}

There are many existing models for social networks and other real-world graphs. We give a short description of some important models; for more details, we recommend the survey of Chakrabarti and Faloutsos \cite{ChFa06}.
Classic examples include
preferential attachment \cite{BaAl99}, small-world models \cite{WaSt98}, copying models
\cite{KuRa+00}, and forest fire \cite{LeKlFa07}. Although these models
may produce heavy-tailed degree distributions, their clustering coefficients of the former
three models are often low \cite{SaCaWiZa10}. Even for models that give high clustering
coefficients, it is difficult to predict their community structure in
advance.  
Because of their unpredictable behavior, it is not possible to match real data with these graphs.
This makes it difficult to validate against real-world interaction networks. 
Moreover, none of these models explain community structure, one of the most
striking features of interaction graphs.

A widely used model
is the Stochastic Kronecker Graph model (known as R-MAT in an early incarnation)
\cite{ChZhFa04,LeChKlFa10}. Notably, it has been selected as the generator for the 
Graph 500 Supercomputer Benchmark \cite{Graph500}. Though it has some desirable
properties \cite{LeChKlFa10}, it can only generate lognormal tails (after suitable
addition of random noise \cite{SePiKo11}) and does not produce high
clustering coefficients~\cite{SaCaWiZa10,PiSeKo11}. 
Multifractal networks are closely related to the SKG model \cite{PaLoVi10}. The random dot product
model \cite{YoSc07, YoSc08} can be made scalable but has never been compared
to real social networks. There have been successful dendogram
based structures that perform community detection and link prediction in real graphs \cite{ClNeMo04,ClMoNe08}. The recent hyperbolic graph model \cite{BoPaKr10, KrPaKi+10} is based on hyperbolic geometry
and has been used to performing Internet routing.

The stochastic block model \cite{BiCh09} has been used to generate
better algorithms for community detection. A degree corrected version \cite{KaNe11} has
been defined to deal with imprecisions in this model. A key
feature of these models is that they break the graph into a \emph{constant} number of relatively large blocks,
and our theory shows that this model does not give a satisfactory explanation
of the clustering coefficients of low degrees (which constitute a majority of the
graph). The LFR community detection benchmark \cite{LaFoRa08} is also somewhat
connected to this model, since it defines a set of communities and has
probabilities of edges within and between these communities.
We stress that these models do not attempt to match real graphs, nor do
they explain the scale-free nature of communities \cite{LeLaDa08,LaKiSa10}.
Our hypothesis and model are very different from these results, because
we use a mathematical formalization to prove the existence of a scale-free
dense ER collection, and the BTER model follows this theory.
Nonetheless, our model can be seen as an extension of these block models, where
the number and sizes of blocks form a scale-free behavior. 
Implicitly, our model can be seen to use a labeling scheme for vertices
that depends on the degrees, and connecting vertices with probabilities
depending on the labels (thereby related to the degree corrected framework of \cite{KaNe11}).

 \begin{figure*}[tb]
   \centering
     \subfloat[Preprocessing: Distribution of nodes into communities]{\label{fig:bin}
   \blackfbox{\includegraphics[width=.2\textwidth,trim=0 0 0 0]{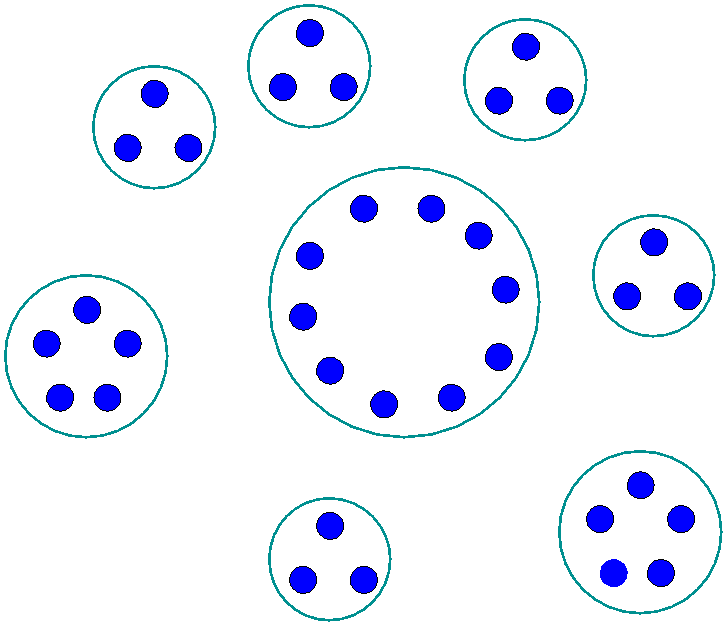}}
   }~~~~~
   \subfloat[Phase 1: Local links within each community]{\label{fig:phase1}
   \blackfbox{\includegraphics[width=.2\textwidth,trim=0 0 0 0]{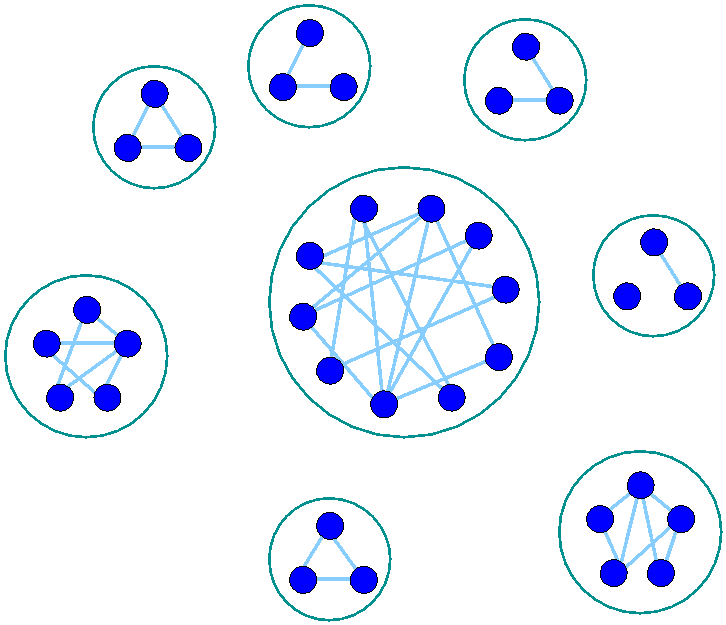}}
   }~~~~~
   \subfloat[Phase 2: Global links across communities]{\label{fig:phase2}
   \blackfbox{\includegraphics[width=.2\textwidth,trim=0 0 0 0]{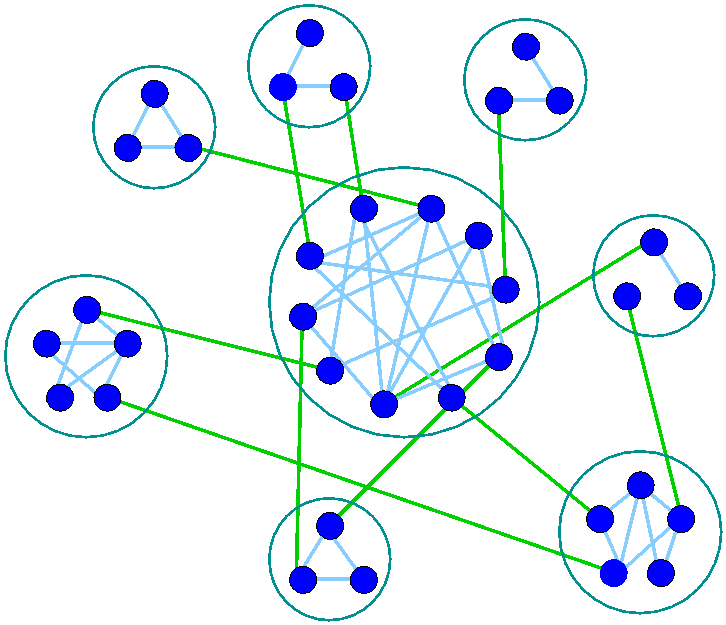}}
   }
   \caption{BTER Model Construction. In the preprocessing phase, the
     nodes are divided into communities. In Phase~1, within-community
     links are generated using the ER model. In Phase~2,
     across-community links are generated using the CL model on the
     \emph{excess} degrees.}
   \label{fig:proc}
 \end{figure*}
\section{Mathematical details} 
\label{sec:theory}

We provide a  sketch of the proof  for \Thm{community};
a complete proof is provided in the supplement.
Our analysis is fundamentally asymptotic, so 
for ease of notation we use the $O(\cdot)$, $\Omega(\cdot)$, and
$\Theta(\cdot)$ to suppress constant factors. 
The notation $A \ll B$ indicates that there exists some absolute
constant $c$ such that $A \leq cB$. 
We let $S$ denote the community of interest and assume that the 
internal degree distribution of the
community $S$ is $\hd_1, \hd_2, \ldots, \hd_r$. We denote 
the number of edges in $S$ by $s = \frac{1}{2}\sum_{i=1}^r \hd_i$.

Based on the given distribution, let $T$ denote the expected number of triangles in $S$.
Since this is a community, we demand that $T$ be at least $\kappa/3$ times the expected
number of wedges, for some constant $\kappa$. This means that 
\begin{equation} \label{eqn:T} 
  T \geq (\kappa/3) \sum_i {\hd_i \choose 2},
\end{equation}
where, for convenience, we define ${d\choose 2} = 0$ when $d=1$.
Let $j$ be the first index such that $\hd_j > 1$. We assume that $\sum_{i > j} \hd^2_j = \Omega(\sum_{i \leq r} \hd^2_i)$,
(which effectively means there are more wedges than degree $1$ vertices).
We can bound ${\hd_i \choose 2} \geq \hd^2_i/4$, and so $T = \Omega(\sum_i \hd^2_i)$.
A key fact we use is the \emph{Kruskal-Katona theorem} \cite{Kru63, Kat68, Fra84}
which states that if a graph has $T$ triangles and $s$ edges, then
$T \leq s^{3/2}$. 
Combining, we have
\begin{equation} \label{eq:kkt}
\sum_{i} \hd^2_i \ll s^{3/2}.
\end{equation}
Now, let us count the expected number
of triangles based on the CL distribution.
For any triple $(i,j,k)$, let $X_{ijk}$ be the indicator random variable
for $(i,j,k)$ being a triangle. This occurs when all the edges $(i,j)$, $(j,k)$, and $(k,i)$
are present, and by independence, this probability is $p_{ij}p_{jk}p_{ki}$,
where $p_{ij} = \hd_i\hd_j/2m$.
The expected number of triangles $T$ can be expressed as $\EX[\sum_{i<j<k} X_{ijk}]$,
which (by linearity of expectation) is $\sum_{i<j<k} \EX[X_{ijk}]$.
Therefore,
\begin{equation} 
T = \sum_{i < j < k} \frac{\hd_i\hd_j}{2s} \cdot \frac{\hd_j\hd_k}{2s} \cdot \frac{\hd_i\hd_k}{2s} 
\leq \frac{\left(\sum_i \hd^2_i\right)^3}{8s^3} \label{eq:tr}
\end{equation}
We argued earlier that $T = \Omega(\sum_i \hd^2_i)$. We can put this bound
in \Eqn{tr} and rearrange to get
\begin{equation} \label{eq:comm}
s^{3/2}  \ll \sum_{i} \hd^2_i
\end{equation}
This is the exact reverse of \Eqn{kkt}! This means that these
quantities are the \emph{same} up to
constant factors. When can this be satisfied? If the community consists of $\sqrt{s}$ vertices all with degree $\sqrt{s}$,
then $\sum_i \hd^2_i = \sum_i s = s^{3/2}$, and the conditions are exactly satisfied. 
Intuitively, to
satisfy both \Eqn{kkt} and \Eqn{comm}, there have to be $\Theta(\sqrt{s})$ vertices of degree $\Theta(\sqrt{s})$.
These vertices form a dense ER graph within the community proving that each community involves
a constant fraction of the edges in an ER graph. 

\section{The BTER Model in Detail} 
\label{sec:model}

The BTER model comprises an interconnected scale-free collection of
communities. 
Intuitively,
short-range connections (Phase 1) tend to be dense and lead to large clustering
coefficients. Long-range connections (Phase 2) are sparse and lead to
heavy-tailed degree distributions. 
We describe the steps in detail below.

\paragraph{Preprocessing}

In the preprocessing step, each node of degree 2 or higher is assigned to a
community. We assume the desired degree distribution $\{d_i\}$ is
given where $d_i$ denotes the desired degree of node $i$.
Roughly speaking, $d$ vertices of degree close to $d$ are assigned to a community
(though in reality, it is somewhat involved than that because of high
degree vertices.) The vertices
are all partitioned into these communities, which has a scale-free
behavior.
Since the degree distribution is an
input to the model, this step is relatively straightforward and
results in a structure as shown in \Fig{bin}. We let $\mathcal{G}_k$
denote the $k$th community and $k_i$ denote the community assignment
for node $i$. 

\GraphFig{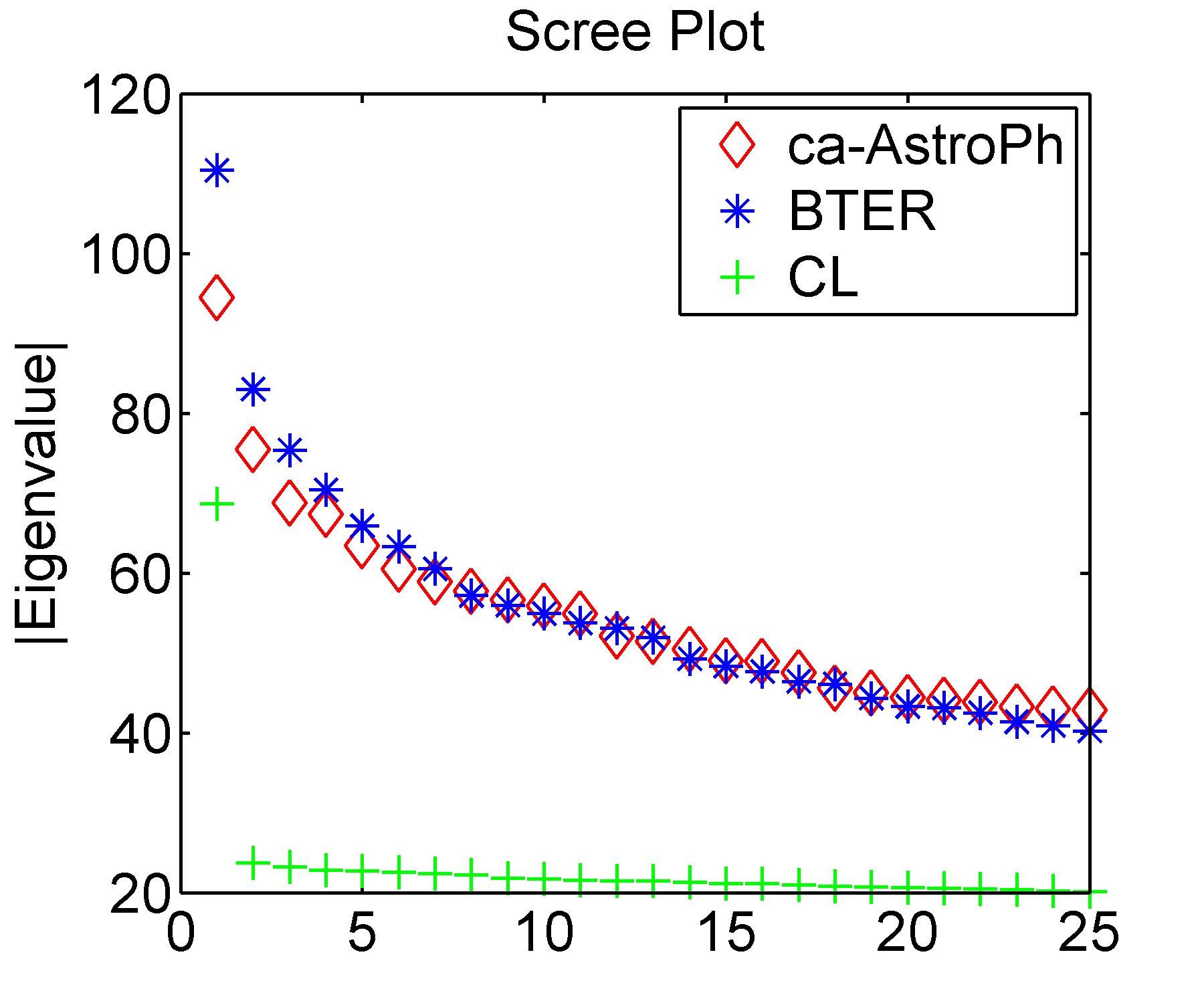}{a co-authorship network from astrophysics
  papers}{Observe the close match of the clustering coefficients of
  the real data and BTER, in contrast to CL. Additionally, the
  eigenvalues of the BTER adjacency matrix are close to those of
  the real data.}  
\GraphFig{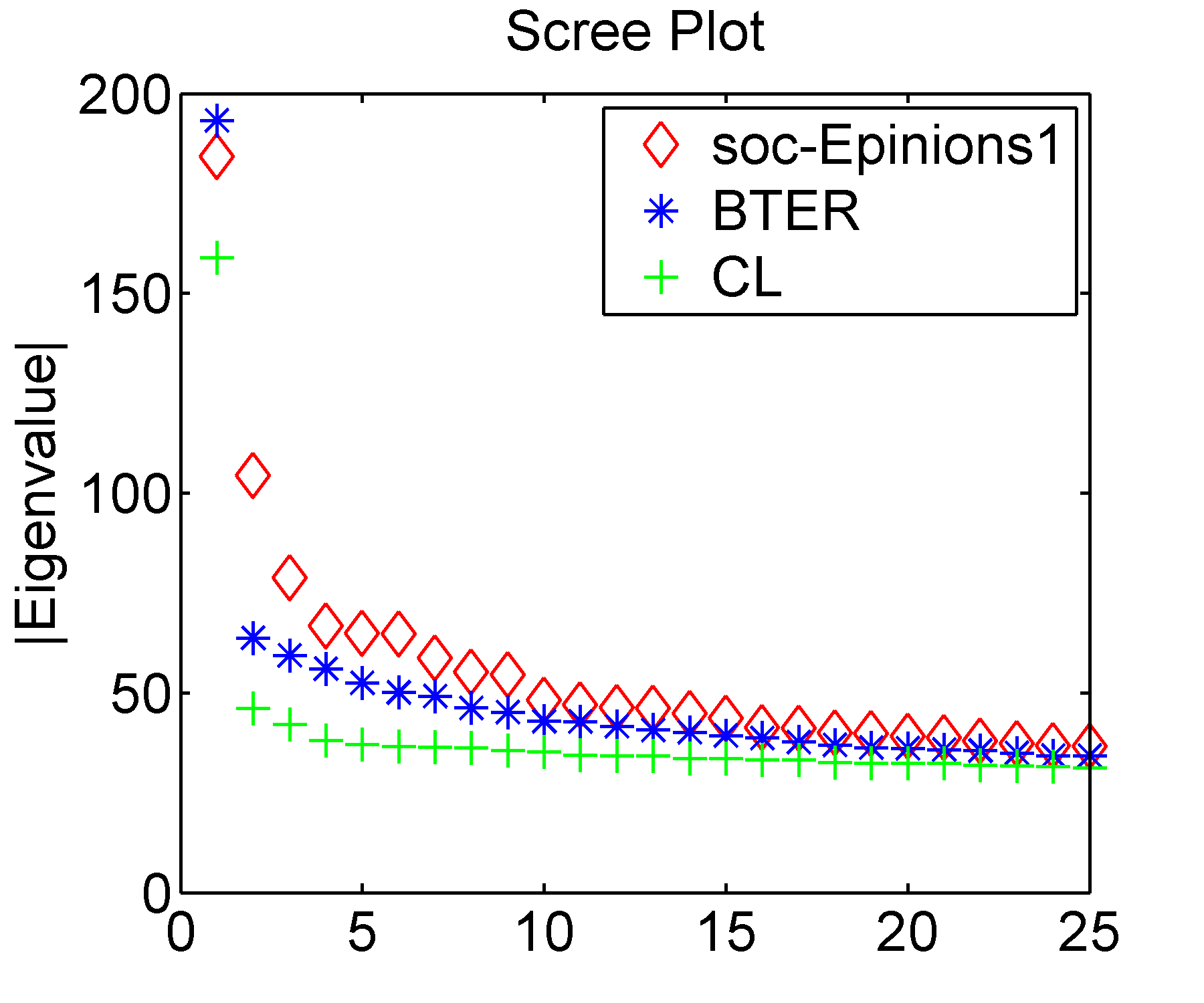}{a social network from the Epinions
  website}{In this case, the clustering coefficients are much smaller
  overall, but the BTER model is still a closer match to the real data
  than CL in terms of both the clustering coefficient and the
  eigenvalues of the adjacency matrix.}

\paragraph{Phase 1}
The local community structure is modeled as an ER graph on each
community. This is illustrated in \Fig{phase1}.
The connectivity of each community is a parameter of the
model. By observing the clustering coefficient plots for real graphs,
we can see that low degree vertices have a much higher clustering coefficient
than higher degree ones. This suggests that small communities are much more tightly
connected than larger ones, and so we adjust the connectivity accordingly.
Any formula may be used; we have found empirically that the following  works
well in practice. We let the edge probability for community
$k$ be defined as
\begin{equation}
  \label{eq:rho}
  \rho_{k} = \rho 
  \left[ 
    1 - \eta
    \left( 
      \frac{\log(\bar d_k +1)}{\log(d_{\max}+1)}
    \right)^2 
  \right],
\end{equation}
where $\bar d_k = \min \Set{d_i|i\in \mathcal{G}_k}$,
$d_{\max}$ is the maximum degree of any node in the entire graph, and
and $\rho$ and $\eta$ are parameters that can be selected for the
best fit to a particular graph. (These were selected by manual
experimentation for our results, but more elaborate procedures could
certainly be developed.) 

\paragraph{Phase 2}
The global structure is determined by interconnecting the
communities. We apply a CL model to the \emph{excess} degree, $e_i$, of each
node, which is computed as follows:
\begin{equation}
  \label{eq:e}
  e_i =
  \begin{cases}
    1, & \text{if } d_i = 1, \\
    d_i - \rho_{k_i} (|\mathcal{G}_{k_i}|-1), & \text{otherwise},
  \end{cases}
\end{equation}
where $|\mathcal{G}_{k}|$ is the size of community $k$.  Given the
$e_i$'s for all nodes, edges are generated by choosing two endpoints
at random. Specifically, the probability of selecting node $i$ is
$e_i/\sum_j e_j$. It is possible to produce duplicate links or
self-links, but these are discarded.  Phase 2 is illustrated in
\Fig{phase2}.

\paragraph{Reference implementation}
A MATLAB reference implementation of BTER 
is included in the supplementary material,  including scripts to
reproduce the findings in this paper.
In this implementation, we have taken some care to reduce the
variance in the CL model with respect to degree-one nodes.
We also generate extra edges in Phase 2 to account for expected
repeats and self-loops that are removed.
These details are described in detail in the supplementary materials.

\section{Results}
\label{sec:results}

We consider comparisons of the BTER model with four real-world data
sets from the SNAP collection~\cite{SNAP}. All the graphs are treated
as undirected.  Properties of these data sets are shown in \Tab{obs}.
We compare BTER with the real data as well as the corresponding CL
model. 
\Fig{ca-AstroPh} shows results on a collaboration network on
124 months of data from the astrophysics (ASTRO-PH) section of the
arXiv preprint server. Here, the edge probabilities in the communities
are given by \Eqn{rho} with $\rho=0.95$ and $\eta=0.05$.
\Fig{soc-Epinions1} shows results on a who-trusts-whom online social
(review) network from the Epinions website. 
Here, the edge probabilities in the communities are
given by \Eqn{rho} with $\rho=0.70$ and $\eta=1.25$.  
Comparisons on two additional datasets listed in the table are
provided in the supplement. 

In the leftmost plots of \Fig{ca-AstroPh} and \Fig{soc-Epinions1}, we
see the comparison of the degree distributions. The degree
distribution for ca-AstroPh has a slight ``kink'' mid-way and does not
conform to any standard degree distribution such as lognormal or power
law. Nonetheless, both BTER and CL are able to match it. The degree
distribution for soc-Epinions is fairly close to a powerlaw, and 
matched well by both BTER and CL. In fact, these models can match any
degree distribution.

The difference between BTER and CL is highlighted  when we instead
consider the clustering coefficient, shown in the center plots of
\Fig{ca-AstroPh} and \Fig{soc-Epinions1}. 
As noted previously, CL cannot have a high clustering coefficient and a
heavy tail, and this is evident in these examples. BTER, on the other
hand, has a close match with the observed clustering
coefficients. The dense ER graphs ensure that all nodes have high
clustering coefficient. 

The importance of matching the clustering coefficients becomes
apparent when considering other features of the graph such as the
eigenvalues of the adjacency matrix, as shown in the rightmost plots
of \Fig{ca-AstroPh} and \Fig{soc-Epinions1}. For ca-AstroPh, the BTER
eigenvalues are a much closer match than the CL eigenvalues because
the community behavior is significant ($C=0.32$). For soc-Epinions1, the
difference between the models in terms of the eigenvalues is less
dramatic because the community behavior is much less evident
($C=0.07$); nonetheless, BTER is still a closer match.

\section{Discussion}
\label{sec:discussion}

We define a community to be a subgraph that is internally well-modeled by CL
(and thus has no further substructure) and highly interconnected (so
that it has many triangles). We prove that any community must contain
a dense ER subgraph.
Therefore, any graph model that captures community structure must contain
dense substructures in the form of dense ER graphs. 
This observation leads naturally to the BTER model, which explicitly
builds communities of varying sizes and simultaneously generates a
heavy tail. 

Fitting the BTER model to real-world data is straightforward.
The community sizes and composition in BTER are determined
automatically according to the degree distribution. We currently use a
simplistic procedure that assumes that all nodes
in the same community have the same expected degree. 
Undoubtedly, this is an unrealistic assumption, but the variance of the
model ensures that the degrees within a community vary considerably
and Phase 2 adds connections between nodes of widely varying degrees.
The connectivity of each ER block is a user-tunable parameter that can
be adjusted to fit observed data. We
currently prescribe a simple formula \Eqn{rho} and fit by trial and
error, but the procedure could certainly
be automated. Moreover, there is no particular requirement
that $\rho_k$ be exactly the same for all communities with the same
$\bar d_k$ (minimum degree) nor that $\rho_k$ be computed by a
deterministic formula.

Our experimental results show that BTER has properties that are
remarkably similar to real-world data sets. We contend that this makes
BTER an appropriate model to use for testing algorithms and
architectures designed for interaction graphs.
In fact, BTER is even
designed to be scalable.
In particular, in Phase 2 we could compute the exact excess degree and use a
matching procedure to complete the graph. The advantage of computing
the excess degree in expectation is that it is more easily parallelized.
In that case, the assignment to
communities, the community connectivity, and the expected excess degree can all
be computed in the preprocessing stage. Both Phase 1 \emph{and} Phase
2 edges can be efficiently generated in parallel via a randomized
procedure. Therefore, the BTER model is suitable for massive-scale
modeling, such as that needed by Graph 500 \cite{Graph500}. 
The details of this implementation are outside the scope of the
current discussion but will be considered in future work.

Our formalism captures the more
advanced notion of \emph{link communities} \cite{AhBaLe10}
(where edges, rather than vertices, form communities). This allows
vertices to participate in many communities. The notion of
communities uses modules over \emph{internal degrees}, so one can
easily imagine a vertex in many communities. \Thm{community} is still
true, and we still get a scale-free collection of ER graphs which
may share vertices. We believe
it is an interesting direction to extend BTER to link (and hence
overlapping) communities.

~~

\begin{table}[h]
\centering
\caption{Data sets for empirical validation}
\label{tab:obs}
\addtolength{\tabcolsep}{3mm}
\begin{tabular}{l|ccc} 
 & Vertices & Edges & $C$ \\\hline 
ca-AstroPh \cite{LeKlFa07}   & 18,772 &  396,100 & 0.32\ \\ 
soc-Epinions1 \cite{RiAgDo03} & 75,879  & 811,480 & 0.07\\
cit-HepPh \cite{GeGiKl03}     & 34,546  & 841,754& 0.15 \\
ca-CondMat \cite{LeKlFa07} &  23,133  & 186,878 & 0.26 \\
\hline
\end{tabular}
\end{table}

\section*{Acknowledgments}
This work was funded by the applied mathematics program at the United
States Department of Energy
and by an Early Career Award from the Laboratory
Directed Research \& Development (LDRD) program at Sandia National
Laboratories. Sandia National Laboratories is
a multiprogram laboratory operated by Sandia
Corporation, a wholly owned subsidiary of Lockheed Martin Corporation,
for the United States Department of Energy's National Nuclear Security
Administration under contract DE-AC04-94AL85000. 


\clearpage

\appendix
\section{Supporting Material}

\subsection{Theoretical details}

The aim of this section is to prove \Thm{community}, which we restate (in slightly different wording)
for convenience. We set ${d\choose 2} = 0$, when $d = 1$. We will use small Greek
letters for constants less than $1$, and small Roman letters for constants who values
may exceed $1$. All constants considered are positive. We will make no attempt
to optimize various constant factors in the proof. The proof is asymptotic
in $s$, the number of edges of our community. That means that the proof
holds for any sufficiently large~$s$.

\begin{theorem} \label{thm:comm-real} Consider a CL graph with 
degree sequence $\hd_1 \leq \hd_2 \leq \cdots \leq \hd_r$
and set $s = \sum_i \hd_i/2$. The quantities $c > 0$ and $\kappa \in (0,1)$ are constants (independent
of $s$).

Let $j$ be the smallest index such that $\hd_j > 1$, and assume that 
$\sum_{i > j} \hd^2_i \geq cj$. Suppose the expected number of triangles
generated with this degree sequence is at least $(\kappa/3)\sum_i {\hd_i \choose 2}$. 
Then (for sufficiently large $s$),
there exists a set of indices $S \subseteq \{1,\ldots,r\}$, such that $|S| = \Omega(\sqrt{s})$
and $\forall k \in S, \hd_k = \Omega(\sqrt{s})$.

(The constants hidden in the $\Omega(\cdot)$ notation only hide a dependence on $c$ and $\kappa$.)
\end{theorem}

The proof of this theorem requires some extremal combinatorics and probability theory. 
We will first state some of these building blocks before describing the main proof.
Henceforth, the assumptions stated in the theorem hold.
An important tool is the Kruskal-Katona theorem that gives an upper bound on the
number of triangles in a graph with a fixed number of edges.

\begin{theorem}[Kruskal-Katona \cite{Kru63, Kat68, Fra84}] \label{thm:kk} If a graph has $t$ triangles and $m$ edges, then $t \leq m^{3/2}$.
\end{theorem}

Since we are dealing with a graph distribution, we need some bounds on the expected
number of edges.

\begin{claim} \label{clm:dev} Let $E$ denote be the random variable denoting the
number of edges in the CL graph defined by $\{\hd_i\}$. Then $\EX[E^{3/2}] \leq 2s^{3/2}$.
\end{claim}

\begin{proof} Let $X_{ij}$ be the indicator random variable for the edge $(i,j)$ being
present. Then, $E = \sum_{i < j} X_{ij}$. By construction, $\EX[E] \leq s$ (we get an inequality
because of possible self-loops). This is the sum of ${n\choose 2}$ independent random variables. Applying
a multiplicative Chernoff bound (Theorems 4.1 and 4.3 of \cite{MoRa01}),
$$ \Pr[E \geq 2\EX[E]] \leq \exp(-\EX[E]/3) $$
Hence, the probability that $E \geq 2s$ is at most $\exp(-s/3)$. Let $\cE$
denote the event that $E \geq 2s$. We can trivially bound $E$ by $n^2 \leq s^2$.
Using Bayes' rule,
\begin{eqnarray*} \EX[E^{3/2}] & = & \Pr(\overline{\cE}) \EX_{\overline{\cE}}[E^{3/2}] + 
\Pr(\cE) \EX_{\cE}[E^{3/2}] \\
& \leq & (2s)^{3/2} + \exp(-s/3)(s^2)^{3/2} \leq 2s^{3/2} 
\end{eqnarray*}
\end{proof}

We now prove some claims about the expected number of triangles and the degree sequence.

\begin{claim} \label{clm:triangles} Let $T$ denote the expected number of triangles.
The constants $\beta$ and $c'$ depend only on $c$ and $\kappa$.
\begin{enumerate}
	\item $T \geq \beta \sum_i \hd^2_i$.
	\item $\sum_i \hd^2_i \leq c's^{3/2}$
\end{enumerate}
\end{claim}

\begin{proof} By assumptions in \Thm{comm-real}, $T \geq (\kappa/3) \sum_i {\hd_i \choose 2}$.
For $\hd_i > 1$, ${\hd_i \choose 2} \geq \hd^2_i/4$ (for large $\hd_i$, it is actually 
much closer to $\hd^2_i/2$).
Hence, $T \geq (\kappa/12) \sum_{i > j} \hd^2_i$, where $j$ is the smallest index of a non-degree
$1$ vertex (as stated in \Thm{comm-real}). By assumption, 
$\sum_{i\leq j} \hd^2_i = j \leq (1/c)\sum_{i > j} \hd^2_i$, and
$\sum_i \hd^2_i \leq (1+1/c)\sum_{i > j} \hd^2_i$. We can
complete the proof of the first part with the following.
$$ T \geq (\kappa/12) \sum_{i > j} \hd^2_i \geq (\kappa/12)(1+1/c)^{-1} \sum_i \hd^2_i$$
Suppose we generate a random CL graph. Let $t$ be the number of triangles
and $E$ be the number of edges (both random variables). By \Thm{kk}, 
$t \leq E^{3/2}$. Taking expectations and applying \Clm{dev}, $T \leq \EX[E^{3/2}] \leq 2s^{3/2}$.
Combining with the first part of this claim, $\sum_i \hd^2_i \leq (2/\beta)s^{3/2}$.
\end{proof}

We come to the proof of the main theorem.

\begin{proof} (of \Thm{comm-real})
We choose $b$ to be a sufficiently large constant, and $\gamma$ to be sufficiently small.
Let $\ell$ be the smallest index such that $\hd_\ell > b\sqrt{s}$.
For a triple of vertices $(i,j,k)$, let $X_{ijk}$ be the indicator random variable
for $(i,j,k)$ forming a triangle. Note that $T = \EX[\sum_{i < j < k} X_{ijk}]$.
Then we have,
\begin{eqnarray} & & \EX[\sum_{i < j < k} X_{ijk}] \nonumber \\
& = & \sum_{i < j < k} \EX[X_{ijk}] \nonumber \\
& \leq & \sum_{i < j < k} \min\left(\frac{\hd_i \hd_j}{2s},1\right) \times \min\left(\frac{\hd_i \hd_k}{2s},1\right) \times \min\left(\frac{\hd_j \hd_k}{2s},1\right) \nonumber\\
& \leq & \sum_{i < j < k} \frac{\hd_i \hd_j}{2s} \times \frac{\hd_i \hd_k}{2s}
\times \min\left(\frac{\hd_j \hd_k}{2s},1\right) \nonumber\\
& \leq & \sum_i \hd^2_i \left[ \sum_{\substack{j \\ j < k \leq \ell}} \frac{\hd^2_j\hd^2_k}{8s^3} + \sum_{j,k \geq \ell} \frac{\hd_j\hd_k}{4s^2} \right] \nonumber\\
& \leq & (\sum_i \hd^2_i)\left[\frac{(\sum_j \hd^2_j)(\sum_{k \leq \ell} \hd^2_k)}{8s^3} + \frac{(\sum_{j \geq \ell} \hd_j)^2}{4s^2} \right]. \nonumber
\end{eqnarray}
By the first part of \Clm{triangles}, $T \geq \beta \sum_i \hd^2_i$. For convenience,
we will replace all the independent indices above by $i$.
\begin{equation} \label{eq:cc-bound} \beta \leq \frac{(\sum_i \hd^2_i)(\sum_{i \leq \ell} \hd^2_i)}{8s^3} + \frac{(\sum_{i \geq \ell} \hd_i)^2}{4s^2}
\end{equation}
By \Clm{triangles}, $\sum_i \hd^2_i \leq c' s^{3/2}$. Furthermore, 
$\sum_i \hd^2_i \geq b\sqrt{s} \sum_{i \geq \ell} \hd_i$.
Combining the two bounds,
$\sum_{i \geq \ell} \hd_i \leq (c'/b) s$.
Applying these bounds in \Eqn{cc-bound} and setting constant $\tau$ appropriately,
\begin{eqnarray*} 
& & \beta \leq \frac{c'\sum_{i \leq \ell} \hd^2_i}{8s^{3/2}} + (c'/2b)^2 \\
& \Longrightarrow & \sum_{i \leq \ell} \hd^2_i \geq (8/c')(\beta-(c'/2b)^2)s^{3/2}  = \tau s^{3/2}
\end{eqnarray*}
(By setting $b$ to be sufficiently large, we can ensure that $\tau$ is a positive
constant.)
Let $\ell'$ be the smallest index such that $\hd_{\ell'} \geq \gamma \sqrt{m}$
and set $s' = \sum_{\ell' \leq i \leq \ell} \hd_i$. 
\begin{eqnarray*} & & \tau s^{3/2} \leq \sum_{i \leq \ell} d^2_i \leq \gamma(s-s')\sqrt{s} + bs'\sqrt{s} \\
& \Longrightarrow & \tau s \leq s'(b-\gamma) + \gamma s \\
& \Longrightarrow & s' \geq s(\tau - \gamma)/(b-\gamma) = \Omega(s)
\end{eqnarray*}
(Again, a sufficiently small $\gamma$ ensures positivity.)
This means that the vertices with indices in $[\ell',\ell]$ are totally
incident to at least $\Omega(m)$ edges. All these vertices have degrees
that are $\Theta(\sqrt{m})$, and hence there are $\Theta(\sqrt{m})$ such
vertices.
\end{proof}

\subsection{Implementation details}

We give some specifics of the BTER implementation, to accompany the
included MATLAB implementation.

In Phase 1, the last ``community'' generally has fewer than $\bar d_k$
nodes because we have run out of nodes. We have found it convenient to
set $\rho_k=0$ for the last community since it is generally pretty
small in any case.

We split the calculation of the Phase 2 edges into three
subphases so that we can specially handle the degree-1 edges. 
The variance for degree-1 vertices in the CL model is high, so we
set aside a proportion of these vertices to be handled ``manually.''
Let $r$ denote the number of degree-1 vertices and assumed the
vertices and indexed from least degree to greatest.
By default, 75\% of the degree-1 vertices are handled ``manually'' (the
exact proportion is user-definable); let $p = \nint{0.75 r}$ denote
this quantity where $\nint{\cdot}$ denotes nearest integer.
We update $e_i$ as follows:
\begin{displaymath}
  e_i \gets
  \begin{cases}
    0, & \text{for } 1 \leq i \leq p,\\
    1.10, & \text{for } p+1 \leq i \leq r,\\
    e_i, & \text{otherwise}.
  \end{cases}
\end{displaymath}
This update removes the first $p$ nodes from the CL part and also
slightly raises the probability of an edge for the remaining $r-p$
degree-1 nodes. This modifications help to balance out the fact that
some higher degree nodes (in expectation) which actually become
degree-1 nodes in the final graph, so we need some of the degree-1
nodes (in expectation) to become higher degree in the final graph.

In Phase 2a, we set
aside $q \leq p$ ($q$ even) degree-1 vertices to be connected to other
degree-1 vertices. This value can be specified by the user or defaults
to
\begin{displaymath}
  q = 2 \; \left\lfloor \frac{p^2}{2 \sum_i d_i} \right\rceil,
\end{displaymath}
which is the expected number of degree-1-to-degree-1 edges expected in
the CL model. This can be accomplished by randomly pairing the selected
vertices. In all of our experiments, we use $q=0$.

In Phase 2b, we manually connect the remaining $(p-q)$ vertices to the
rest of the graph. For each degree-1 vertex, we select an endpoint
proportional to $e_i$.

In Phase 2c, we finally create the CL model. We modify the expected
degrees to account for the edges used in Phase 2b and to account for
duplicates. Thus, we update $e_i \gets \eta e_i$ where
\begin{displaymath}
  \eta = 1 - 2 \frac{p-q}{p-q+\sum_i e_i} + \beta,
\end{displaymath}
where $\beta$ is the proportion of duplicates. We use $\beta = 0.10$
in our experiments. The total number of edges generated in Phase 2c
(including repeats and self-edges, which are discarded) is
$\nint{\sum_i e_i/2}$.

\subsection{Additional experimental results}

We consider two additional experiments, as shown
in \Fig{cit-HepPh} and \Fig{ca-CondMat}.
The graph \Fig{cit-HepPh}  represents    
 all the citations within a dataset on the high energy physics phenomenology section of the arXiv preprint server. 
For \Fig{cit-HepPh}, we use an alternate formula for $\rho_k$ as follows:
\begin{displaymath}
  \rho_{k} = 0.7
  \left[ 
    1 - 0.6
    \left( 
      \frac{\log(\bar d_k +1)}{\log(d_{\max}+1)}
    \right)^3
  \right].
\end{displaymath}
And  \Fig{ca-CondMat} shows results on a collaboration network on
124 months of data from the condensed matter (COND-MAT) section of the
arXiv server. Here, the edge probabilities in the communities
are given by \Eqn{rho} with $\rho=0.95$ and $\eta=0.95$.

\GraphFig{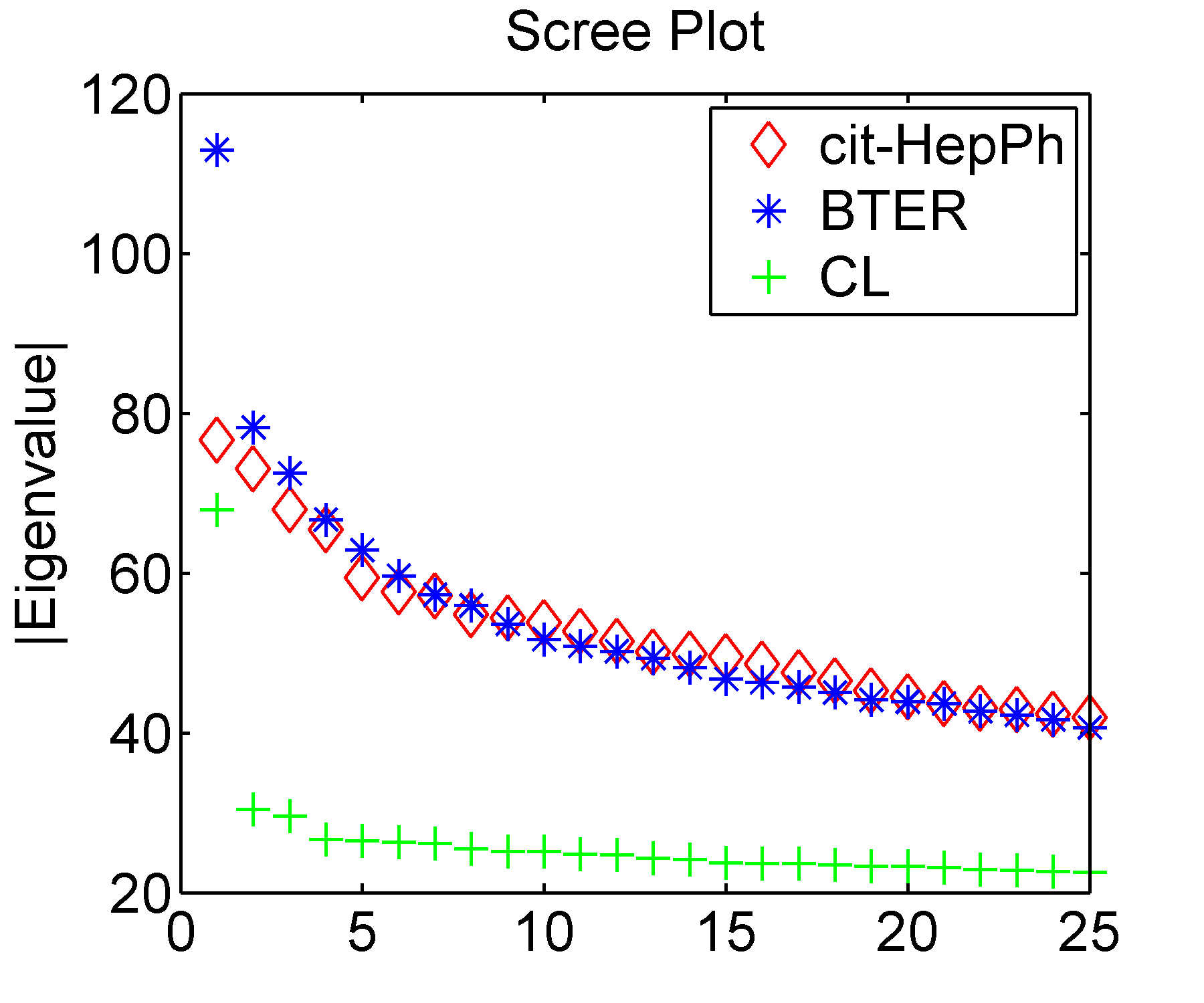}{a citation network of High Energy Physics
  papers}{}

\GraphFig{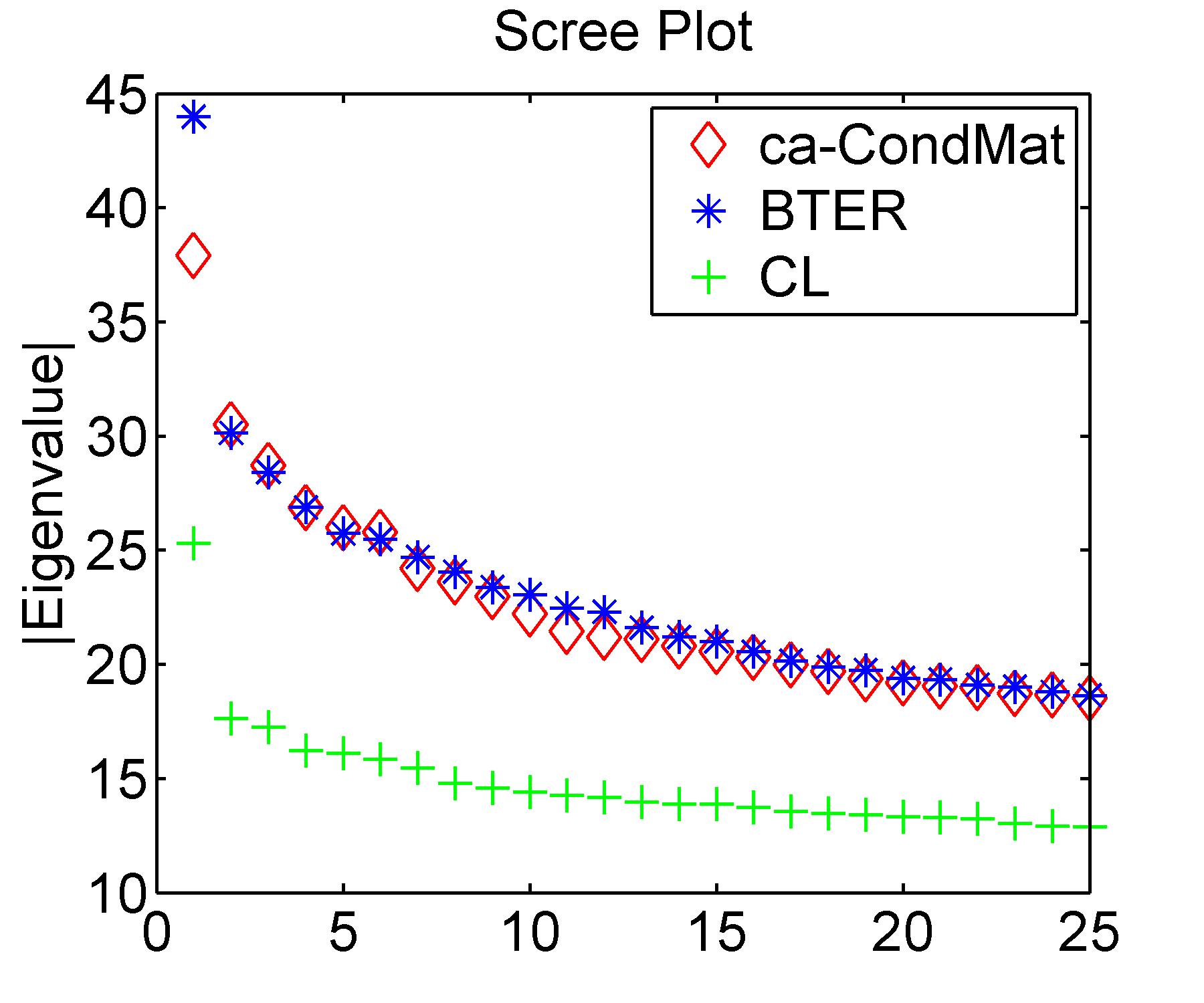}{a co-authorship network of Condensed Matter
  physics}{}

The results on these two graphs are consistent with with our earlier results. The leftmost plots show that both CL and BTER can match the degree distributions of  the original graph, as expected. Again, the clustering coefficients plots in the middle highlight the strengthen of BTER, and how it differs from CL: BTER matches the clustering coefficients closely, while 
CL does not produce any significant number of triangles.  The
rightmost column shows that the eigenvalues of the adjacency matrices
of BTER are closer to those of the original graph than those produced
by CL.

\subsection{Codes and data}

The codes and data are available at \url{http://www.sandia.gov/~tgkolda/bter_supplement/}.

\end{document}